\documentclass[preprint,11pt]{article}
\usepackage{amsmath,amssymb,amsthm}
\usepackage{color}
\newcommand{\F}{\mathbb{F}}
\newtheorem{definition}{Definition}
\newtheorem{lemma}{Lemma}
\newtheorem{theorem}{Theorem}

\newtheorem{corollary}{Corollary}
\newtheorem{proposition}{Proposition}

\begin{document}

\title{
Some Results on Bent-Negabent Boolean Functions over Finite Fields\thanks{This is a sufficiently revised and extended version of the paper \cite{Sseta12}. Section \ref{new} is a completely new contribution.}
}
\author{Sumanta Sarkar\\
Centre of Excellence in Cryptology\\
Indian Statistical Institute, Kolkata, INDIA\\
Sumanta.Sarkar@gmail.com
}

\date{}
\maketitle

\begin{abstract}
We consider negabent Boolean functions that have Trace representation.
We completely characterize quadratic negabent monomial functions.
We show the relation between negabent functions and bent functions
via a quadratic function.
Using this characterization, we give infinite classes of bent-negabent
Boolean functions over the finite field $\F_{2^n}$, with the maximum possible degree, $n \over 2$.
These are the first ever constructions of negabent functions with trace representation that have 
optimal degree.

\end{abstract}


\section{Introduction}
Hadamard-Walsh transform is an important tool in characterizing Boolean functions.
For example, many cryptographic properties can be analyzed by the Hadamard-Walsh transform.
A function on even number of variables that has the maximum possible distance from the affine 
functions is called a bent function. These functions have equal absolute spectral values under the
Hadamard-Walsh transform and was first introduced by Rothaus \cite{RO76}.
It is natural to investigate the spectral values of Boolean functions under some other
Fourier transform. 
In 2007, Parker and Pott \cite{PP07}, considered the nega-Hadamard transform.
and introduced negabent functions. 
These functions have equal absolute spectral values under the nega-Hadamard transform.
The periodic autocorrelation values of a bent function are all zero. 
The negaperiodic autocorrelation value of a Boolean function under the nega-Hadamard transform 
is the analogue of the periodic autocorrelation value. 
The negaperiodic autocorrelation values are all zero for a negabent function.
These properties of a negabent function motivate us to study it further.
Negabent functions which are also bent are interesting as they 
have extreme properties in terms of two different Fourier transforms.

Results on negabent functions can be found in \cite{PP07,SPP08,P00,RP06,Sugatanega,Sindocrypt09}. 
As an example, the $6$-variable function
$
x_4(x_1 x_2 \oplus x_2 x_3 \oplus x_1 \oplus x_2) \oplus
x_5(x_1 x_2 \oplus x_2 x_3 + x_3) + x_6(x_1 \oplus x_3)$
is a cubic negabent function.

In \cite{PP07,SPP08} some classes of Boolean functions which are both
bent and negabent (bent-negabent) have been identified. In \cite{SPP08},
construction of negabent functions has been shown in the class of
Maiorana-McFarland bent functions. 
It is interesting to note that all 
the affine functions (both odd and even variables) are negabent
\cite[Proposition 1]{PP07}. 
In \cite{Sindocrypt09}, symmetric negabent
functions have been characterized and shown to be all affine for both
odd and even number of variables.
The maximum degree of an $n$-variable bent-negabent function is 
$n \over 2$.
Very recently construction of bent-negabent functions have been given in
\cite{Pottnega} with the optimal degree.

In this paper, we characterize
the negabent functions which are defined over finite fields, 
{\em i.e.}, functions with Trace representation.

Let $\mathbb{F}_2^n$ be the vector space formed by the binary $n$-tuples
and $\mathbb{F}_{2^n}$ be the finite field with $2^n$ elements. For a set $E$,
the set of non zero elements of $E$ is denoted by $E^*$.

In \cite{PP07}, quadratic negabent Boolean functions defined over the vector space
$\mathbb F_2^n$  were characterized. 
Any quadratic Boolean function can be written as 
\begin{eqnarray*}
g(x_1, \ldots, x_n)& = & \sum_{1 \leq i < j \leq n} q_{i, j} x_i x_j
+ \sum_{1\leq i \leq n} l_i x_i + c\\
\ & = & x Q x^T + L x^T + c,
\end{eqnarray*}
where $Q=(q_{i,j})$ is an upper triangular binary matrix,
$L = (l_1, \ldots, l_n)$ is a binary vector and $c\in \{0, 1\}$.
Consider the binary symmetric matrix $B = Q + Q^T$ with the zero diagonal,
which is the symplectic matrix corresponding to the quadratic function $g$.
It is well known that a quadratic function $g$ is bent if and only if
the corresponding matrix $B$ has full rank. In \cite{PP07}, it was proved
that a quadratic function $g$ is negabent if and only if the matrix $B+I$
has full rank, where $B$ is the corresponding symplectic matrix and
$I$ is the identity matrix.

In this paper, first we consider quadratic monomials
defined over the field $\mathbb F_{2^n}$ and is of the form
\begin{equation}
\label{form-quad}
f: x \mapsto Tr_1^n(\lambda x^{2^k+1}), ~~~~\mbox{where } \lambda \in \mathbb F_{2^n}^*.
\end{equation}

The $\lambda$'s for which $f$ is bent is well known.
We characterize those $\lambda$'s for which $f$ is negabent.
We also give characterization of $\lambda$ for even $n$ such that $f$ is bent-negabent. 
The existence of quadratic bent-negabent functions is known \cite{PP07}. We reprove
the existence by simple counting argument and using the characterization of 
quadratic bent-negabent monomials.

We also study the negabent property of Maiorana-McFarland bent functions
$f: \mathbb F_{2^t}\times \mathbb F_{2^t} \mapsto \mathbb F_2$
defined by $f(x,y) = Tr_1^t(x \pi(y) + h(y))$, where $\pi(y)$ is a permutation
polynomial over $\mathbb{F}_{2^t}$ and $h(y)$ is any polynomial over $\F_{2^t}$.
We present a necessary and sufficient condition such that these functions are
negabent. As a consequence, we show that when the permutation $\pi$ is 
$x \mapsto x^{2^i}$, the bent function $f$ is negabent if and only if
$h(y)$ is bent.
From this we show how a Maiorana-McFarland bent-negabent function of degree $n \over 4$
over $\mathbb{F}_{2^n}$ can be obtained.

Then we show that given a bent function $f$ over $\F_{2^n}$,
it is possible to obtain a negabent function by adding a quadratic function, and
vice versa. Using this result we are able to present infinite classes of bent-negabent
functions having the optimal degree $n \over 2$.

\section{Preliminary}
An $n$-variable Boolean function is a mapping 
$f : \mathbb{F}_2^n \mapsto \mathbb{F}_2$.
The Hamming weight of a binary string $S$ is the number of $1$'s
in $S$ and it is denoted as $wt(S)$.
An $n$-variable Boolean function $f$ can be written as a function of
$x_1, \ldots, x_n$ variables as follows,
$$f(x_1,x_2,\ldots , x_n) =
\bigoplus_{a = (a_1, \ldots, a_n) \in \mathbb{F}_2^n} \mu_a
(\prod_{i=1}^n x_i^{a_i}), \mbox{ where } \mu_a \in \mathbb{F}_2.$$
This is called the algebraic normal form (ANF) of $f$.
The degree, $\deg(f)$, of $f$ is defined as
$\displaystyle\max _{a \in \mathbb{F}_2^n}\{wt(a) | \mu _a \ne 0\}$.

Let $\lambda = (\lambda_1, \ldots, \lambda_n)$ and
$x = (x_1, \ldots, x_n)$ be two vectors in $\mathbb{F}_2^n$
and $\lambda \cdot x = \lambda_1 x_1 \oplus \ldots \oplus \lambda_n x_n$.
Then the Hadamard-Walsh transform value of $f$ at
$\lambda$ is given by
\begin{equation}
\label{eqnhw}
\mathcal{H}_f(\lambda) = \frac{1}{2^\frac{n}{2}} \sum_{x \in \mathbb{F}_2^n} (-1)^{f(x) \oplus \lambda \cdot x}.
\end{equation}
The function is called bent if $|H_f(\lambda)| = 1$ for all $\lambda \in \mathbb{F}_2^n$.
After the introduction of bent functions in \cite{RO76}, there have been many 
constructions of bent functions, for instance, \cite{Dillon}, \cite{Carlet93}, \cite{LHTK13}, and references therein.

For $a \in \mathbb{F}_2^n$, the periodic autocorrelation value of $f$
is computed as 
$$\tau_a = \sum_{x \in \mathbb{F}_2^n}(-1)^{f(x) \oplus f(x \oplus a)}.$$
A Boolean function $f$ is bent if and only if $\tau_a = 0$ for all
$a \in {\mathbb{F}_2^n}^*$. 

The nega-Hadamard transform value of 
$f$ at $\lambda \in \mathbb{F}_2^n$
is given by
\begin{equation}
\label{eqnnh}
\mathcal{N}_f(\lambda) = 
\frac{1}{2^{\frac{n}{2}}} \sum_{x \in \mathbb{F}_2^n} (-1)^{f(x) \oplus \lambda \cdot x}I^{wt(x)},
\end{equation}
where $I = \sqrt{-1}$, is the imaginary unit of the complex number.
Note that $\mathcal{N}_f(\lambda)$ is complex valued.

The function $f : \mathbb F_{2}^n \rightarrow \mathbb F_2$ is called {\em negabent} \cite{PP07} if
the magnitude of $\mathcal{N}_f(\lambda)$ is $1$, {\em i.e.},
$|\mathcal{N}_f(\lambda)| = 1$ for all $\lambda \in \mathbb F_2^n$.
In the following theorem we state an alternate characterization
of negabent functions in terms of their negaperiodic autocorrelation values
which has been shown in \cite[Theorem 2]{PP07} and \cite[Lemma 3]{Sugatanega}.

\begin{theorem}
A Boolean function $f$ is negabent if and only if
\begin{equation}
\sum_{x \in \mathbb F_2^n}
(-1)^{f(x) \oplus f(x \oplus y)}(-1)^{x \cdot y} = 0
\end{equation}
for all $y \in {\mathbb F_2^n}^*$.
\end{theorem}

From this theorem, we see that the correlation values between the function
$f(x)$ and $f(x \oplus y) \oplus y\cdot x$ are all zero 
for all $y \in {\mathbb F_2^{n}}^*$. Moreover,
for even number of variables, if a negabent function is also a bent function,
then the correlation values of the function $f(x)$ and $f(x \oplus y)$ are also
equal to zero for all $y \in {\mathbb F_2^{n}}^*$. 
Therefore, the functions which are both bent and negabent
are interesting to study. We call these functions \emph{bent-negabent}.
Bent functions can exist only on even number of variables and it has degree more than
$1$. However, all the affine functions are negabent \cite{PP07} which tells
that negabent functions exist for both even and odd number of variables.

\section{Characterization of negabent functions over the finite field $\mathbb F_{2^n}$}
Now we consider Boolean functions defined over the field
$\mathbb F_{2^n}$ and we characterize the negabent property of those functions.
The vector space $\mathbb F_2^n$ can be easily identified with the field
$\mathbb F_{2^n}$ by choosing a basis of $\mathbb F_{2^n}$ over $\mathbb F_2$.
The function
$Tr_1^n : \mathbb F_{2^n} \mapsto \mathbb F_2$ is defined
as $$Tr_1^n(x) = x + x^2 + \ldots + x^{2^{n-1}}.$$
We denote $Tr_1^n$ simply by $Tr$ and $``+"$ is the finite field addition.
If we choose the basis
$\{\alpha_1, \ldots, \alpha_n\}$ to be self dual then it can be shown that
$Tr(xy) = \sum_{i = 1}^n x_i y_i$.

Henceforth, in this paper we choose the basis to be self dual.
It is also notable that a linear function over $\mathbb F_{2^n}$
is given by $\ell(x) = Tr(ax)$, $a \ne 0$.

Given a polynomial $F(x)$ over $\F_{2^n}$, we can get a Boolean function
$f : \F_{2^n} \rightarrow \F_2$ defined as $f(x) = Tr(F(x))$.
The highest binary weight of the exponents of $F(x)$ is denoted as the algebraic degree of $F(x)$,
then the degree of $f(x) = Tr(F(x))$ is equal to the algebraic degree of $F(x)$.

With the above discussions it is now clear that we can characterize negabent
functions defined over the finite field as follows.
\begin{proposition}
The function $f : \mathbb F_{2^n} \rightarrow \mathbb F_2$ is negabent if and only if
\begin{equation}
\label{eq-main-nega}
\sum_{x \in \mathbb F_{2^n}} (-1)^{f(x) + f(x+a) + Tr(ax)} = 0
\end{equation}
for all $a \in \mathbb F_{2^n}^*$.
\end{proposition}

It is known that a bent function $f : \mathbb{F}_{2^n} \mapsto \mathbb{F}_{2}$
is bent if and only if

\begin{equation}
\label{eq-main-bent}
\sum_{x \in \mathbb F_{2^n}} (-1)^{f(x) + f(x+a)} = 0
\end{equation}
for all $a \in \mathbb F_{2^n}^*$.

A function for which both (\ref{eq-main-nega}) and (\ref{eq-main-bent})
hold is a bent-negabent function.

An $n$-variable Boolean function $\phi$ is called balanced if its weight is $2^{n-1}$.
Note that $\sum_{x \in \mathbb F_{2^n}} (-1)^{\phi(x)} = 0$ if and only if
$\phi$ is balanced.
Therefore, if $f$ is bent-negabent, then both
$f(x) + f(x+a)$ and $f(x) + f(x+a) + Tr(ax)$ are balanced
for all $a \in \mathbb F_{2^n}^*$.

\subsection{Linear structure of negabent functions}
As the sum (\ref{eq-main-nega}) involves derivatives of $f$, {\em i.e.}, $f(x) + f(x+a)$,
here we briefly discuss about the linear structures of negabent functions.
\begin{definition}
An $a \in \mathbb{F}_{2^n}^*$ is said to be a linear structure of a polynomial
$F(x)$ over $\mathbb{F}_{2^n}$ if the derivative $F(x) + F(x+a)$ is constant.
\end{definition}

From (\ref{eq-main-bent}), it is clear that a bent function can not have a linear structure.
However, a negabent function can have a linear structure. In fact, if $f$ is such that
any $a \in \mathbb{F}_{2^n}^*$ is a linear structure, then $f$ is negabent.
As in that case, the term $f(x)+ f(x+a)+ ax$ is $c+ax$, for some constant $c$, 
which is affine, {\em i.e.}, balanced.
This happens when $f(x)$ is an affine polynomial, 
which proves that if $f$ is affine it is negabent.

A polynomial $F(x)$ over $\F_{2^n}$ is called a complete mapping polynomial
if both $F(x)$ and $F(x) + x$ are permutation polynomials.
We use such permutation polynomials in our constructions.

\subsection{Quadratic negabent monomial functions}
We consider quadratic monomials and characterize when they are negabent.
\begin{proposition}
\label{quadprop}
Let $f: \mathbb F_{2^n} \rightarrow \mathbb F_2$ be a quadratic function of the form
$$f(x) = Tr(\lambda x^{2^k+1}).$$ 
Then $f$ is negabent if and only if
\begin{equation}
\label{ne-qu-co1}
\lambda^{2^{n-k}} a^{2^{n-k}} + \lambda a^{2^k} + a \ne 0
\end{equation}
for all $a \in \mathbb F_{2^n}^*$.
\end{proposition}
\begin{proof}
From (\ref{eq-main-nega}) we know that the function $f$ is negabent 
if and only if for all $a \in \mathbb F_{2^n}^*$,
\begin{eqnarray}
\label{ne-qu-co}
\nonumber
\sum_{x \in \mathbb{F}_{2^n}}
(-1)^{Tr(\lambda x^{2^k+1}) + Tr(\lambda(x + a)^{2^k+1})
		+ Tr(ax)} &=& 0, \\
\mbox{ {\em i.e.}, } \sum_{x \in \mathbb{F}_{2^n}}
(-1)^{Tr\left((\lambda^{2^{n-k}} a^{2^{n-k}} + \lambda a^{2^k}
		+ a)x \right)} &=& 0.
\end{eqnarray}

Note that (\ref{ne-qu-co}) is true if and only if 
$$\lambda^{2^{n-k}} a^{2^{n-k}} + \lambda a^{2^k} + a \ne 0,$$
for all $a \in \mathbb F_{2^n}^*$. Hence the result.

\end{proof}

\begin{proposition}
\label{prop-nega-perm}
The quadratic function $f: \mathbb F_{2^n} \rightarrow \mathbb F_2$
of the form $f(x) = Tr(\lambda x^{2^k+1})$ 
is negabent if and only if
$$P(x) = \lambda^{2^{n-k}} x^{2^{n-k}} + \lambda x^{2^k} + x$$
is a permutation polynomial over $\mathbb F_{2^n}$.
\end{proposition}
\begin{proof}
From Proposition \ref{quadprop}, we know that $f$ is negabent if and only if
$$\lambda^{2^{n-k}} a^{2^{n-k}} + \lambda a^{2^k} + a \ne 0$$
for all $a \in \mathbb F_{2^n}^*$. That means for such $\lambda \in \mathbb F_{2^n}^*$, 
$P(x) = \lambda^{2^{n-k}} x^{2^{n-k}} + \lambda x^{2^k} + x$
has no non zero root in $\mathbb F_{2^n}$. Note that $P(x)$ is a linearized
polynomial and linearized polynomial is permutation if and only if
it has no non zero root. Hence the result.

\end{proof}

The polynomial $P(x)$ is a permutation over $\mathbb F_{2^n}$ if and only if
the polynomial $P(x)^{2^k}$ is a permutation over $\mathbb F_{2^n}$. Therefore, 
$Tr(\lambda x^{2^k+1})$ is a negabent if and only if 
$\lambda x + \lambda^{2^k} x^{2^{2k}} + x^{2^k}$
is a permutation. There are some results on the number of solutions
of the polynomial 
$\lambda^{2^k} x^{2^{2k}} + x^{2^k} + \lambda x$
in \cite{HK10}.

Let us point out some results related to the equation
$$\lambda^{2^k} x^{2^{2k}} + x^{2^k} + \lambda x = 0$$
which has been extensively studied in \cite{HK10}.
Let $\gcd(k, n) = d \geq 1$ and $n = td$ for $t > 1$. 
A particular sequence of polynomials over $\mathbb F_{2^n}$ 
is introduced as follows.
\begin{eqnarray}
\nonumber
C_1(x) &=& 1,\\
\nonumber
C_2(x) &=& 1,\\
C_{i+2}(x) &=& C_{i+1}(x) + x^{2^{ik}}C_i(x)
\mbox{~~for~~} 1 \leq i \leq t-1.
\end{eqnarray}

Another polynomial $Z_n(x)$ over $\mathbb F_{2^n}$ is defined as follows.
\begin{eqnarray}
\nonumber
Z_1(x) &=& 1,\\
Z_t(x) &=& C_{n+1}(x) + x C_{t-1}^{2^k}(x)
\mbox{~~for~~} t > 1.
\end{eqnarray}

Then we have the following result from \cite[Proposition 2]{HK10}.
\begin{proposition}
\label{z-root}
Let $\gcd(k, n) = d \geq 1$ and $n = td$ for $t > 1$. 
The equation
$$\lambda^{2^k} x^{2^{2k}} + x^{2^k} + \lambda x = 0$$
defined over $\mathbb F_{2^n}$ has no non zero solution in 
$\mathbb F_{2^n}$ if and only if
$Z_t(\lambda) \ne 0$.
\end{proposition}

The form of $\lambda$ for which
$Z_t(\lambda) = 0$ is known, which is as follows.
\begin{lemma}\cite[Corollary 1]{HK10}
\label{root-zt}
Let $\gcd(k, n) = d \geq 1$ and $n = td$ for $t > 1$.
Then $\alpha$ is a zero of $Z_t(x)$ in $\mathbb F_{2^n}$ 
if and only if it is of the form
\begin{equation}
\label{zero-form}
\frac{v_0^{2^{2k}+1}}{(v_0+v_1)^{2^k+1}},
\end{equation}
where $v_0 \in \mathbb F_{2^n} \setminus \mathbb F_{2^d}$ and $v_1 = v_0^{2^k}$. 
The total number of distinct roots are
$$\begin{cases}
\frac{2^{n+d} - 2^d}{2^{2d} - 1} \mbox{~~~~ for even } t\\
\frac{2^{n+d} - 2^{2d}}{2^{2d} - 1} \mbox{~~~ for odd } t.
\end{cases}$$
\end{lemma}

Therefore, we have the following theorem which characterizes the quadratic 
negabent monomials.
\begin{theorem}
The function $f : x \mapsto Tr(\lambda x^{2^k+1})$ is negabent if and 
only if $\lambda$ can not be written as 
$\frac{v_0^{2^{2k}+1}}{(v_0+v_1)^{2^k+1}}$ for
$v_0 \in \mathbb F_{2^n} \setminus \mathbb F_{2^d}$ and $v_1 = v_0^{2^k}$ 
where $\gcd(k, n) = d$ and $n = td$.
\end{theorem}
\begin{proof}
Proposition \ref{prop-nega-perm} and Proposition \ref{z-root} 
imply that $f$ is negabent if and only if $\lambda$ is not a zero of
$Z_t(x)$ where $\gcd(k,n) = d$ and $n = td$. 
From Lemma \ref{root-zt}, we know that $Z_t(\lambda) \ne 0$ if and only if 
$\lambda$ is not of the form 
$\frac{v_0^{2^{2k}+1}}{(v_0+v_1)^{2^k+1}}$
where $v_0 \in \mathbb F_{2^n} \setminus \mathbb F_{2^d}$ and $v_1 = v_0^{2^k}$. 

\end{proof}

\subsection{Quadratic bent-negabent functions}
We recall the well known result on the quadratic bent monomials.
This is directly taken from \cite{DL04}.
\begin{lemma}\cite{DL04}
\label{gold-bent}
Let $\lambda \in \mathbb F_{2^n}$ and $n$ even.
The function 
$f : \mathbb F_{2^n} \rightarrow \mathbb F_2$
with $$f(x) = Tr(\lambda x^{2^k+1})$$
is bent if and only if 
$$\lambda \notin \{x^{2^k+1} \mid x \in \mathbb F_{2^n}\}.$$
\end{lemma}

Note that $\lambda^{2^{n-k}} x^{2^{n-k}} + \lambda x^{2^k}$
is a permutation if and only if
$\lambda \notin \{x^{2^k+1} \mid x \in \mathbb F_{2^n}\}$.
Therefore, if $f(x) = Tr(\lambda x^{2^k+1})$
is negabent, then Proposition \ref{prop-nega-perm} tells that
$\lambda^{2^{n-k}} x^{2^{n-k}} + \lambda x^{2^k} + x$
is also a permutation polynomial, {\em i.e.}, 
$\lambda^{2^{n-k}} x^{2^{n-k}} + \lambda x^{2^k}$
is a complete mapping polynomial.
We summarize these results as follows. 

\begin{theorem}
\label{th-bent-negabent}
Let $\lambda \in \mathbb F_{2^n}$ where $n$ is even.
The function 
$f : \mathbb F_{2^n} \rightarrow \mathbb F_2$
with $$f(x) = Tr(\lambda x^{2^k+1})$$
is bent negabent if and only if one of the following two 
equivalent statements holds.
\begin{enumerate} 
\item
$\lambda^{2^{n-k}} x^{2^{n-k}} + \lambda x^{2^k}$
is a complete mapping polynomial.
\item
$\lambda$ is neither of the form
$\frac{v_0^{2^{2k}+1}}{(v_0+v_1)^{2^k+1}}$ nor of the form
$v^{2^k+1}$ for 
$v \in \mathbb F_{2^n}$,
$v_0 \in \mathbb F_{2^n} \setminus \mathbb F_{2^d}$ and $v_1 = v_0^{2^k}$,
where $\gcd(k, n) = d$ and $n = td$.
\end{enumerate}
\end{theorem}

The existence of quadratic bent-negabent functions is known 
\cite[Theorem 5]{PP07}.
However, we reprove the same result by simple counting argument
and using the previous characterization of the bent-negabent functions.
\begin{theorem}
For all $n \geq 4$, quadratic bent-negabent functions always exist.
\end{theorem}
\begin{proof} We show that there always exists a $\lambda \in \mathbb F_{2^n}$
which satisfies the condition 2 of Theorem \ref{th-bent-negabent}.

If $\gcd(2^k+1, 2^n-1) = 1$, then $x \mapsto x^{2^k+1}$
is a bijection. Then for any $\lambda \in \mathbb F_{2^n}$ there exists
$x \in \mathbb F_{2^n}$ such that $\lambda = x^{2^k+1}$. Therefore, if 
$Tr(\lambda x^{2^k+1})$ is bent then $\gcd(2^k+1, 2^n-1) > 1$.
Since $2$ does not divide both of $2^k+1$ and $2^n-1$. Therefore,
$\gcd(2^k+1, 2^n-1) \geq 3$. 
Let $S_1 = \{x^{2^k+1} | x\in \mathbb F_{2^n}\},$ then $|S_1| \leq \frac{2^n-1}{3}$.
On the other hand, if $\gcd(k, n) = d$ and $n = td$, then
by Lemma \ref{root-zt}, we know that the number of possible
$\lambda \in \mathbb F_{2^n}$ such that $\lambda$ is of the form
$\frac{v_0^{2^{2k}+1}}{(v_0+v_1)^{2^k+1}}$ is
$\begin{cases}
\frac{2^{n+d} - 2^d}{2^{2d} - 1} \mbox{ for even } t\\
\frac{2^{n+d} - 2^{2d}}{2^{2d} - 1} \mbox{ for odd } t
\end{cases}$.
Let $S_2 = \{y \in \mathbb F_{2^n} | y = \frac{v_0^{2^{2k}+1}}{(v_0+v_1)^{2^k+1}}\}$.
Note that $|S_1 \cup S_2| = |S_1| + |S_2| - |S_1 \cap S_2|$. Then
$$|S_1 \cup S_2| \leq  \frac{2^n-1}{3} 
+ \frac{2^{n+d}-2^{d}}{2^{2d}-1} - |S_1 \cap S_2|.$$ 
Therefore,
\begin{eqnarray*}
2^n - |S_1 \cup S_2| & \geq & 2^n - \frac{2^n-1}{3} 
- \frac{2^{n+d} - 2^d}{2^{2d} - 1} + |S_1 \cap S_2|\\
\ & = & (2^n -1).\frac{2.2^{2d} - 3.2^d - 2}{3(2^{2d}-1)}
+ |S_1 \cap S_2| + 1\\
\ & \geq & |S_1 \cap S_2| + 1, ~~~~\mbox{ since } 2.2^{2d} - 3.2^d - 2 \geq 0.
\end{eqnarray*}

Therefore, we see that a $\lambda \in \mathbb F_{2^n}$ always exists
that satisfies Condition 2 of Theorem \ref{th-bent-negabent}.

This proves the theorem.

\end{proof}

Below we characterize bent-negabent functions when $n = 2k$.
\begin{proposition} 
Let $n = 2k$ and $f: x \mapsto Tr(\lambda x^{2^k+1})$ be a quadratic function
defined over $\mathbb F_{2^n}$.
Then $f$ is negabent if and only if
$\lambda + \lambda^{2^k} \ne 1$. Moreover, $f$ is bent-negabent
if and only if 
$\lambda + \lambda^{2^k} \notin \mathbb F_2$.
\end{proposition}
\begin{proof}
By Proposition \ref{prop-nega-perm} we have that
$f$ is negabent if and only if
$P(x) = \lambda^{2^{n-k}} x^{2^{n-k}} + \lambda x^{2^k} + x$ is a permutation,
{\em i.e.}, $P(x)^{2^k} = \lambda x + \lambda^{2^k} x^{2^{2k}} + x^{2^k}$ 
is a permutation. Since $n = 2k$, therefore,
$P(x)^{2^k} = (\lambda + \lambda^{2^k})x + x^{2^k}$. Now
$(\lambda + \lambda^{2^k})x + x^{2^k}$ is permutation if and only if
$(\lambda + \lambda^{2^k})x + x^{2^k} \ne 0$,
{\em i.e.}, $\lambda + \lambda^{2^k} \ne x^{2^k-1}$, for all $x \in \mathbb F_{2^n}^*$.
Note that $\lambda + \lambda^{2^k} \in \mathbb F_{2^k}$ for all 
$\lambda \in \mathbb F^{2^n}$ and the mapping
$\lambda \mapsto \lambda + \lambda^{2^k}$ is onto.
Let us consider the group
$G = \{x^{2^k-1}| x \in \mathbb F_{2^n}^*\}$. The intersection
of $\mathbb F_{2^k}$ and $G$ is $\{1\}$.
Therefore, $Tr(\lambda x^{2^k+1})$ is negabent if and only if
$\lambda + \lambda^{2^k} \ne 1$.
We know that $f$ is bent if and only if $\lambda \ne x^{2^k+1}$
for some $x \in \mathbb F_{2^n}$. Note that if $\lambda = x^{2^k+1}$
then $\lambda \in \mathbb F_{2^k}$ and $\lambda + \lambda^{2^k} = 0$.
Therefore, $f$ is bent-negabent if and only if 
$\lambda + \lambda^{2^k} \notin \mathbb F_{2}$.

\end{proof}

\section{Maiorana-McFarland bent-negabent functions}
Maiorana-McFarland is an important class of bent functions
which was extensively studied by Dillon \cite[pp.~90-95]{Dillon}. This
class is usually called the {\em class $\cal M$} of bent functions.
\begin{lemma}\label{mmf}
Let $n=2t$. Let us consider a Boolean function $f$ defined by
\begin{equation}\label{eq:mmf}
f~:~(x,y)\in \mathbb{F}_{2^t} \times \mathbb{F}_{2^t}
  ~~\mapsto ~~Tr^t_1\left(x\pi(y)+h(y)\right)
\end{equation}
where $\pi$ is a function over $\mathbb{F}_{2^t}$ and $h$ is any function
on $\mathbb{F}_{2^t}$. Then $f$ is a bent function if and only if
$\pi$ is a bijection.
\end{lemma}

\begin{theorem}
\label{th:mmf}
Let $f$ be a Maiorana-McFarland function as in Lemma \ref{mmf}.
Then $f$ is negabent if and only if for all $a, b \in \mathbb F_{2^t}^*$ 
\begin{equation}
\sum_{y \in Y_{a,b}} (-1)^{Tr_1^t(a \pi(y)) + h(y) + h(y+b) +by)} = 0,
\end{equation}
where $Y_{a,b} = \{ y \in \mathbb F_{2^t} | \pi(y) + \pi(y+b) = a\}$ such that
$Y_{a, b}$ is non empty.
\end{theorem}
\begin{proof}
From (\ref{eq-main-nega}) we have, $f(x, y)$ is negabent
if and only if for all $(a, b) \in \mathbb F_{2^t} \times \mathbb F_{2^t} \setminus \{(0,0)\}$
\begin{eqnarray}
\label{eq:mmf-nega}
\nonumber
\sum_{(x, y) \in \mathbb F_{2^t} \times \mathbb F_{2^t}}
(-1)^{f(x,y) + f(x+a, y+b) + Tr_1^t(ax) + Tr_1^t(by)} & = & 0\\
\nonumber
\sum_{(x, y) \in \mathbb F_{2^t} \times \mathbb F_{2^t}}
(-1)^{Tr_1^t(x(\pi(y) + \pi(y+b)+a)) + Tr_1^t(a \pi(y+b) + h(y) + h(y+b) + by)} & = & 0.
\end{eqnarray}
Let 
$$S_{a,b} = \displaystyle \sum_{(x, y) \in \mathbb F_{2^t} \times \mathbb F_{2^t}}
(-1)^{Tr_1^t(x(\pi(y) + \pi(y+b)+a)) + Tr_1^t(a \pi(y+b) + h(y) + h(y+b) + by)}.$$
We treat the sum $S_{a,b}$ in the following cases.

\noindent
{\bf CASE 1:} For $a \ne 0$ and $b = 0$.
Then
\begin{eqnarray*}
S_{a,b} & = & \sum_{(x, y) \in \mathbb F_{2^t} \times \mathbb F_{2^t}}
(-1)^{Tr_1^t(ax) + Tr_1^t(a \pi(y))}\\
\ & = & \sum_{x \in \mathbb F_{2^t}} (-1)^{Tr_1^t(ax)}
\sum_{y \in \mathbb F_{2^t}} (-1)^{Tr_1^t(a\pi(y))}\\
\ & = & 0. 
\end{eqnarray*}

\noindent
{\bf CASE 2:} For $a = 0$ and $b \ne 0$.
Then
\begin{eqnarray*}
S_{a,b} & = & \sum_{(x, y) \in \mathbb F_{2^t} \times \mathbb F_{2^t}}
(-1)^{Tr_1^t(x(\pi(y) + \pi(y+b))) + Tr_1^t(h(y) + h(y+b) + by)}\\
\ & = & \sum_{y \in \mathbb F_{2^t}} (-1)^{Tr_1^t(h(y) + h(y+b) + by)}
\sum_{x \in \mathbb F_{2^t}} (-1)^{Tr_1^t(x(\pi(y) + \pi(y+b)))}\\
\ & = & \sum_{y \in \mathbb F_{2^t}} (-1)^{Tr_1^t(h(y) + h(y+b) + by)} \times 0\\
\ & \ & ~~~~~~~~~~~~~\mbox{ since $\pi$ is a permutation, } \pi(y) \ne \pi(y+b)\\
\ & = & 0.
\end{eqnarray*}

\noindent
{\bf CASE 3:} For $a \ne 0$ and $b \ne 0$.
Then
\begin{eqnarray*}
S_{a,b} & = & \sum_{(x, y) \in \mathbb F_{2^t} \times \mathbb F_{2^t}}
(-1)^{Tr_1^t(a \pi(y+b) + h(y) + h(y+b) + by) + Tr_1^t(x(\pi(y) + \pi(y+b)+a))}\\
\ & = & \sum_{y \in \mathbb F_{2^t}} (-1)^{Tr_1^t(a \pi(y+b) + h(y) + h(y+b) + by)}
\sum_{x \in \mathbb F_{2^t}} (-1)^{Tr_1^t(x(\pi(y) + \pi(y+b)+a))}.
\end{eqnarray*}
If there exists some $y$ such that $y \notin Y_{a,b}$,
{\em i.e.}, $\pi(y) + \pi(y+b) \ne a$,
then $$\sum_{x \in \mathbb F_{2^t}} (-1)^{Tr_1^t(x(\pi(y) + \pi(y+b)+a))} = 0.$$
On the other hand if $y \in Y_{a,b}$, {\em i.e.}, $\pi(y) + \pi(y+b) = a$,
then $$\sum_{x \in \mathbb F_{2^t}} (-1)^{Tr_1^t(x(\pi(y) + \pi(y+b)+a))} = 2^t.$$
Therefore,
\begin{eqnarray*}
S_{a,b} & = & 2^t\sum_{y \in Y_{a,b}} (-1)^{Tr_1^t(a \pi(y+b) + h(y) + h(y+b) + by)}\\
\ & = & 2^t\sum_{y \in Y_{a,b}} (-1)^{Tr_1^t(a \pi(y) + a^2 + h(y) + h(y+b) + by)}\\
\ & \ & ~~~~~~~~~~~~ \mbox{since $\pi(y+b) = \pi(y) + a$ for $y \in Y_{a,b}$}\\
\ & = & 2^t (-1)^{Tr(a^2)} \sum_{y \in Y_{a,b}}(-1)^{Tr_1^t(a \pi(y) + h(y) + h(y+b) + by)}.
\end{eqnarray*}

Therefore, $S_{a,b} = 0$ if and only if
$$\displaystyle\sum_{y \in Y_{a,b}}(-1)^{Tr_1^t(a \pi(y) + h(y) + h(y+b) + by)} = 0.$$

Thus after discussing all the above cases it is clear that the Maiorana-McFarland 
bent function $f$ is negabent if and only if
$$\sum_{y \in Y_{a,b}}(-1)^{Tr_1^t(a \pi(y) + h(y) + h(y+b) + by)} = 0.$$

\end{proof}

This Theorem gives us the clue to construct negabent functions over the finite
fields that belong to the class of Maiorana-McFarland bent functions.

\begin{definition}
A mapping $F : \mathbb{F}_{2^n} \rightarrow \mathbb{F}_{2^n}$ is
called homomorphic if $F(x + y) = F(x) + F(y)$ and $F(xy) = F(x) F(y)$
for all $x, y \in \mathbb{F}_{2^n}$.
\end{definition}
The only possible homomorphic permutation over $\mathbb{F}_{2^n}$
is of the form $x \mapsto x^{2^i}$. Note that $Tr_1^n(x) = Tr_1^n(x^{2^i})$,
therefore the mapping $x \mapsto Tr_1^n(x)$ is invariant under the
action of this permutation. 
Using this observation we show an interesting consequence 
of Theorem \ref{th:mmf}, when the permutation $\pi$ is chosen as 
$\pi(x) = x^{2^i}$. 

\begin{theorem}
\label{homo}
Let $f : (x,y) \in \mathbb F_{2^{t}} \times \mathbb F_{2^{t}} \mapsto
\mathbb{F}_{2}$ be a Maiorana-McFarland bent function
given by
\begin{equation}
f(x,y) = Tr_1^t(x y^{2^i} + h(y)),
\end{equation}
Then $f$ is negabent if and only if $Tr_1^t(h(y))$ is a
bent function over $\mathbb F_{2^t}$.
\end{theorem}
\begin{proof}
Let $\pi(y) = y^{2^i}$. Then $\pi(y)$ is a homomorphic
permutation polynomial over $\mathbb F_{2^{t}}$.
From the linearity of $\pi$ we have 
$\pi(y)+\pi(y+b) = a$ 
if and only if $\pi(b) = a$.
Then \begin{equation*}
Y_{a,b} = \begin{cases}
\mathbb{F}_{2^t} & \text{when} ~\pi(b) = a\\
\text{empty} & \text{when} ~\pi(b) \ne a.
\end{cases}
\end{equation*}
Since $\pi$ is a permutation, for each $a$ there will be a $b$ such that
$\pi(b) = a$. For such $a$ and $b$ 
{\small
\begin{eqnarray*}
\sum_{y \in \mathbb F_{2^t}}(-1)^{Tr_1^t(a \pi(y) + h(y) + h(y+b) + by)}
& = & \sum_{y \in \mathbb F_{2^t}}(-1)^{Tr_1^t(\pi(b)\pi(y) + by
+ h(y) + h(y+b))}\\
& = & \sum_{y \in \mathbb F_{2^t}}(-1)^{Tr_1^t(\pi(by)) + Tr_1^t(by) 
+ Tr_1^t(h(y) + h(y+b))}.
\end{eqnarray*}
}
Note that $Tr_1^t(y) = Tr_1^t(\pi(y))$, for all $y \in \mathbb{F}_2^t$.
So
$$
\sum_{y \in \mathbb F_{2^t}}(-1)^{Tr_1^t(a \pi(y) + h(y) + h(y+b) + by)}
 = \sum_{y \in \mathbb F_{2^t}}(-1)^{Tr_1^t(h(y) + h(y+b))}.
$$
Using Theorem \ref{th:mmf}, the function  $f$ is negabent if and only if
$$\sum_{y \in \mathbb F_{2^t}}(-1)^{Tr_1^t(h(y)) + Tr_1^t(h(y+b))} = 0,$$
for all $b \in \mathbb F_{2^t}^*$, {\em i.e.}, $Tr_1^t(h(y))$ is a bent function
over $\mathbb F_{2^t}$.

Thus the result follows.

\end{proof}

Similar kind of result was proved in \cite{Sugatanega}, where the
function was defined over the vector space $\mathbb{F}_2^n$ and the permutation 
was  such that $wt(x + y) = wt(\pi(x) + \pi(y))$. 
However, the result of Theorem \ref{homo} is quite distinct as it is
in the domain of finite fields.
Moreover, Theorem \ref{th:mmf} is a general characterization of bent-negabent
Maiorana-McFarland functions and several constructions of Maiorana-McFarland
bent-negabent functions can be obtained from this. 
For instance, Theorem \ref{homo} allows us to construct bent-negabent 
Maiorana-McFarland function of degree $n/4$ over $\mathbb{F}_{2^n}$
by choosing a bent function of degree $n/4$ as $h$, where $n = 2t$.

\section{Negabent functions from bent functions}
\label{new}
We show that given a negabent function over a finite field, one can construct a bent function,
and vice versa.
First we define 
$Q : \mathbb F_{2^n} \rightarrow \mathbb F_2$ as
\begin{equation}
\label{def-Q}
Q(x) = \sum_{i = 1}^{{n \over 2} - 1}Tr_1^n(x^{2^i + 1}) + Tr_1^{n \over 2}(x^{2^{n \over 2} + 1})
\end{equation}

As mentioned earlier, for simplicity we write $Tr_1^n(x) = Tr(x)$.

We also mention a result
from \cite{decompose} and \cite{carlet-parseval} which will be useful in
proving our result. These results were proved for Boolean functions defined
over vector spaces, however, it is easy to see the equivalent results
when the Boolean function is defined by the trace representation.

\begin{lemma}{\cite[Corollary 1]{decompose}}
\label{pascale}
Suppose $H_{\beta} = \{x \in \F_{2^n} : Tr(\beta x) = 0\}$ is 
a hyperplane. If $f: \F_{2^n} \rightarrow \F_2$ is bent, then for
any $a \notin H_{\beta}$,
\begin{eqnarray}
\sum_{x \in \mathbb{F}_{2^n}}
(-1)^{f(x) + f(x+a) + Tr(\beta x)} = 0.
\end{eqnarray}
\end{lemma}

\begin{lemma}{\cite[Theorem V.3]{carlet-parseval}}
\label{bent-hyp}
The Boolean function $f: \F_{2^n} \rightarrow \F_2$ is bent if and only 
if there exists a hyperplane $\mathcal{H}$ such that $f(x) + f(x+a)$ is balanced
for every nonzero $a \in \mathcal{H}$. 
\end{lemma}

\vspace{1cm}
\begin{theorem}
\label{bent-quad-th}
Suppose $f : \mathbb F_{2^n} \rightarrow \mathbb F_2$, 
and $Q$ is as defined in (\ref{def-Q}).

\begin{enumerate}

\item
\label{b2n}
if $f$ is bent then $f+Q$ is negabent.

\item
\label{n2b}
If $f$ is negabent then $f+Q$ is bent,
\end{enumerate}
\end{theorem}

\begin{proof}
Suppose $a \in F_{2^n}^*$, then
\begin{eqnarray}
\label{der-Q}
\nonumber
Q(x) + Q(x+a) 	& = & \sum_{i = 1}^{{n \over 2}-1}Tr(a^{2^i} x + a x^{2^i}) + Tr_1^{n \over 2}(a x^{2^{n \over 2}} + a^{2^{n \over 2}} x) + \text{constant} \\
\nonumber
\		& = & \sum_{i = 1}^{{n \over 2}-1}Tr((a^{2^i} + a^{2^{n-i}})x) + Tr_1^n(a^{2^{n \over 2}} x) + \text{constant} \\
\nonumber
\		& = & \sum_{i = 1}^{{n \over 2}-1}Tr((a^{2^i} + a^{2^{n-i}})x) 
			+ Tr(a^{2^{{n \over 2}}}x)+ Tr(ax)\\
\nonumber
		&   &~~~~~~~~~~	  + Tr(ax) +  \text{constant},\\
\nonumber			
\		& = & Tr(Tr(a)x) + Tr(ax) + \text{constant} \\
\		& = & Tr(a)Tr(x) + Tr(ax) + \text{constant}.
\end{eqnarray}

\noindent
Without any loss of generality we ignore the constant term in $Q(x)+Q(x+a)$.

\vspace{1cm}
\noindent
{
\bf Case 1:}
We prove that if $f$ is bent, then $f+Q$ is negabent, for which we have to show that
\begin{center}
$ \begin{array}{l}
\displaystyle
\sum_{x \in \mathbb{F}_{2^n}}
(-1)^{f(x) + f(x+a) +Q(x) + Q(x+a) + Tr(ax)} = 0
\end{array}$
\end{center}

We have
\begin{center}
$ \begin{array}{l}
\displaystyle
\sum_{x \in \mathbb{F}_{2^n}}
(-1)^{f(x) + f(x+a) +Q(x) + Q(x+a) + Tr(ax)} \\
\displaystyle
= \sum_{x \in \mathbb{F}_{2^n}} (-1)^{f(x) + f(x+a)+ Tr(a)Tr(x)}.
\end{array}$
\end{center}

\vspace{1cm}
\noindent
{\bf Subcase 1.1:}
If $Tr(a) = 0$, then
\begin{eqnarray}
\nonumber
\sum_{x \in \mathbb{F}_{2^n}}
(-1)^{f(x) + f(x+a) + Tr(a)Tr(x)} = 0, \mbox{ since $f$ is bent}.
\end{eqnarray}

\vspace{1cm}
\noindent
{\bf Subcase 1.2:}
If $Tr(a) = 1$, then $a$ does not belong to the hyperplane
$$H_1 = \{x \in \F_{2^n} : Tr(1.x) = 0\}.$$ Therefore, by Lemma \ref{pascale},
\begin{eqnarray}
\nonumber
\sum_{x \in \mathbb{F}_{2^n}}  
(-1)^{f(x) + f(x+a) + Tr(1.x)} = 0.
\end{eqnarray}
So for any $a \in F_{2^n}^*$,
\begin{eqnarray}
\nonumber
\sum_{x \in \mathbb{F}_{2^n}}
(-1)^{f(x) + f(x+a) +Q(x) + Q(x+a) + Tr(ax)} = 0.
\end{eqnarray}
This implies that is $f + Q$ is negabent.

\vspace{1cm}
\noindent
{\bf Case 2:}
Next we suppose that $f$ is negabent and prove that $g = f+Q$ is bent.
For any $a \in \F_{2^n}^*$ we have
\begin{center}
$\begin{array}{l}
\displaystyle \sum_{x \in \mathbb{F}_{2^n}}
(-1)^{g(x) + g(x+a)} \\
= \displaystyle \sum_{x \in \mathbb{F}_{2^n}}
(-1)^{f(x) + f(x+a) + Q(x) + Q(x+a)} \\
= \displaystyle \sum_{x \in \mathbb{F}_{2^n}} (-1)^{f(x) + f(x+a)+ Tr(a)Tr(x) + Tr(ax)}, \mbox{by } (\ref{der-Q})\\
= \displaystyle \sum_{x \in \mathbb{F}_{2^n}} (-1)^{f(x) + f(x+a) + Tr(ax)}, \mbox{if } a \in H_1, i.e. Tr(a) = 0\\
= 0, \mbox{since $f$ is negabent}.
\end{array}$
\end{center}

Therefore, we see that for any nonzero $a$ that belongs to the hyperplane $H_1$,
$g(x) + g(x+a)$ is balanced. Hence by Lemma \ref{bent-hyp}, we prove that $g$ is 
bent.

\end{proof}

This theorem has interesting consequences. 

\begin{corollary}
\label{btnbt-cor}
The Boolean function $f: \F_{2^n} \rightarrow \F_2$ is bent-negabent 
if and only if both $f$ and $f+Q$ are bent.
\end{corollary}

\begin{corollary}
The Boolean function $f: \F_{2^n} \rightarrow \F_2$ is bent-negabent 
if and only if $f+Q$ is bent-negabent.
\end{corollary}

\begin{corollary}
The function $Q$ is bent but not negabent.
\end{corollary}
\begin{proof}
It is easy to check that $Q$ is bent by looking at its derivative given
in (\ref{der-Q}).

Now on the contrary, assume that $Q$ is negabent. Then $g = Q+Q = 0$
is bent (by Theorem \ref{bent-quad-th}), which is a contradiction.

\end{proof}

We now use the result of Theorem \ref{bent-quad-th} to construct bent-negabent
functions.
Note that any two quadratic bent functions are affine equivalent.
It is clear that there is one-one correspondence between the
bent function defined over $\F_{2^{2t}}$ and $\F_{2^t} \times \F_{2^t}$.
With abuse of notation we use $Q(x,y)$ defined over $\F_{2^t} \times \F_{2^t}$
 as the corresponding bent function for $Q(x)$ which is defined in (\ref{def-Q}).
That means the bent function 
$G : (x,y) \in \mathbb F_{2^{t}} \times \mathbb F_{2^{t}} \mapsto \mathbb{F}_{2}$
given by
\begin{equation}
\label{G}
G(x,y) = Tr_1^t(x y), 
\end{equation}
is affine equivalent to the bent function $Q(x,y)$. 
This also means by Theorem \ref{bent-quad-th} that if $f(x,y)$ is a bent function
then $f(x,y) + Q(x,y)$ is negabent and vice versa.

Suppose $G(x,y)$ and $Q(x,y)$ are related by the relation
\begin{eqnarray}
\label{q-g}
Q(x,y) &=& G(\alpha_1 x + \alpha_2,\alpha_3 y + \alpha_4)  +
	 Tr_1^t(\beta x) + Tr_1^t(\gamma y) + c,
\end{eqnarray}
for some $\alpha_1, \alpha_2, \alpha_3, \alpha_, \beta, \gamma$ in $\F_2^t$,
$c \in \F_2$.
 
\begin{theorem}
\label{comp-th}
Let $f : (x,y) \in \mathbb F_{2^{t}} \times \mathbb F_{2^{t}} \mapsto
\mathbb{F}_{2}$ be a Maiorana-McFarland bent function
given by
$$
f(x,y) = Tr_1^t(x \pi(y)) + Tr_1^t(h(y)),
$$
where $\pi(y)$ is a complete mapping polynomial over $\F_{2^t}$,
$h(y)$ is any polynomial over $\F_{2^t}$,
and $G : (x,y) \in \mathbb F_{2^{t}} \times \mathbb F_{2^{t}} \mapsto
\mathbb{F}_{2}$ defined by $G(x,y) = Tr_1^t(x y)$.
Then
\begin{eqnarray}
\label{comp-th-cons}
\nonumber
F(x,y) &=& 
f(\alpha_1 x + \alpha_2, \alpha_3 y + \alpha_4) + G(\alpha_1 x + \alpha_2, \alpha_3 y + \alpha_4) \\
	& \ & + Tr_1^t(\beta x) + Tr_1^t(\gamma y) + c 
\end{eqnarray}
is a bent-negabent function.
\end{theorem}

\begin{proof}
We have
\begin{eqnarray}
\nonumber
f(x,y) + G(x,y) 	& = &  Tr_1^t(x \pi(y)) + Tr_1^t(h(y)) +  Tr_1^t(x y) \\
\nonumber
\			& = &  Tr_1^t(x (\pi(y) + y)) + Tr_1^t(h(y)).
\end{eqnarray}
Since $\pi(y)$ is a complete mapping polynomial over $\F_{2^t}$, $(\pi(y) + y)$ is a permutation polynomial,
so $f+G$ is a bent function. This also implies that
$$
F(x,y) = 
f(\alpha_1 x + \alpha_2, \alpha_3 y + \alpha_4)+G(\alpha_1 x + \alpha_2, \alpha_3 y + \alpha_4) + Tr_1^t(\beta x) + 
Tr_1^t(\gamma y) + c 
$$
is a bent function.
We have
\begin{center}
 
$\begin{array}{l}
F(x,y)\\
= f(\alpha_1 x + \alpha_2, \alpha_3 y + \alpha_4)+G(\alpha_1 x + \alpha_2, \alpha_3 y + \alpha_4) + Tr_1^t(\beta x) 
+ Tr_1^t(\gamma y) + c \\
= f(\alpha_1 x + \alpha_2, \alpha_3 y + \alpha_4) + Q(x,y), \text{by} \ref{q-g}.
\end{array}$

\end{center}

Note that $F(x,y) + Q(x,y) = f(\alpha_1 x + \alpha_2, \alpha_3 y + \alpha_4)$ is also bent.
So both $F(x,y)$ and $F(x,y) + Q(x,y)$ are bent.
Therefore, by Corollary \ref{btnbt-cor}, $F(x,y)$ is bent-negabent.

\end{proof}

At this point, we would like to refer to \cite[Theorem 22]{Sugatanega}, which also states
a result that is similar to Theorem \ref{comp-th}. In that result, the Boolean
function is defined over the vector space $\F_2^n$. Note that
complete mapping polynomials are defined over finite fields, however, the proof of \cite[Theorem 22]{Sugatanega}
works in the vector space domain.
They claim that $\pi(x_1, \ldots, x_t)$ is a permutation of $\F_2^t$ that corresponds to
the permutation $\pi(x)$ over the field $\F_{2^t}$, as well as
$\pi(x_1, \ldots, x_t) \oplus (x_1, \ldots, x_t)$ is the permutation
of $\F_2^t$ that corresponds to the permutation $\pi(x)+x$ over the field $\F_{2^t}$.
But it is not clear how this correspondence is realized.
On the other hand, Theorem \ref{comp-th} can directly apply the complete mapping
polynomials as the underlying Boolean function is defined over a finite filed.

Now we construct infinite classes of $n$-variable bent-negabent function with the maximum degree $n \over 2$.
Our construction is similar to that of Theorem $5$ of \cite{Pottnega}. 
Their proof works when there is a permutation polynomial $p(x)$ over the vector space $\F_2^n$ such that $p(x) + x$
is also a permutation polynomial over $\F_2^n$. However,
this kind of permutation over the vector space $\F_2^n$ is not characterized, on the other hand,
these kind of permutation polynomials (complete mapping polynomials) are well characterized over finite field.

\begin{corollary}
Suppose $n = 2t$. Then
the $n$-variable function $F(x,y)$ defined in Theorem \ref{comp-th}, where the polynomial $h(y)$
has algebraic degree $t = {n \over 2}$ is a bent-negabent function
of degree $n \over 2$.
\end{corollary}
\begin{proof}
The algebraic degree of $h(y)$ is $t = {n \over 2}$, which implies that
the degree of $f(x,y) + G(x,y)$ is also $n \over 2$.
That also implies that the degree 
$
F(x,y) = 
f(\alpha_1 x + \alpha_2, \alpha_3 y + \alpha_4)+G(\alpha_1 x + \alpha_2, \alpha_3 y + \alpha_4) + Tr_1^t(\beta x) +
Tr_1^t(\gamma y) + c
$ is $n \over 2$.

\end{proof}

Two infinite classes of complete mapping polynomial are given in \cite{yann}.

\begin{theorem}{\cite[Theorem 4.3]{yann}}
\label{yann-th}
 Let $p$ be a prime and $m$ and $\ell$ are two positive integers. Let $k$
 be the multiplicative order of $p$ in $\mathbb{Z}_m$.
 Assume $a \in \F_{p^{k \ell}}$ is such
 that $(-a)^m \ne 1$. Then the polynomials
 $$\pi_1(x) = x(x^{\frac{p^{k \ell m}-1}{m} } + a),$$
 and
 $$\pi_2(x) = ax^{\frac{p^{k \ell m}-1}{m} + 1},$$
 are complete mapping polynomials over $\F_{p^{k\ell m}}$.
\end{theorem}

\begin{theorem}
Suppose $n = 2 k \ell m$ and $\pi_1(y), \pi_2(y)$ are the complete mapping polynomials
as given in Theorem \ref{yann-th}.
Then the $n$-variable function $F(x,y)$ as given in (\ref{comp-th-cons}) is a bent-negabent
with degree $n \over 2$,
for $f(x,y) = Tr_1^{{n \over 2}}(x \pi_1(y)) + Tr_1^{{n \over 2}}(y^{2^{n \over 2}-1})$
and $f(x,y) = Tr_1^{{n \over 2}}(x \pi_2(y)) + Tr_1^{{n \over 2}}(y^{2^{n \over 2}-1})$.
\end{theorem}
\begin{proof}
 This follows easily as the algebraic degree of $Tr_1^{{n \over 2}}(y^{2^{n \over 2}-1})$ is $n \over 2$.
 
\end{proof}

\section{Conclusion}
We have presented some characterizations of negabent functions over the finite field.
The analysis done here is useful in order to obtain further results
on negabent functions over finite fields.
In this paper, we have characterized quadratic negabent monomials. 
The characterization
of negabent monomials of higher degree will be interesting.
We also have characterized negabent functions which are Maiorana-McFarland bent.
Moreover. we have presented a construction of bent-negabent functions with optimal
degree. This is the second known construction of such functions. However,
it is interesting to see further classes of such functions. 

\section{Acknowledgments}
The author would like to thank Pascale Charpin who helped in proving Case 2 of Theorem \ref{bent-quad-th}.
He is also thankful to Alexander Kholosha for helpful discussions.

\bibliographystyle{alpha}
\bibliography{negabent.bib}

\end{document}